%% file: ICCPaperFinal.tex
\documentclass[10pt,conference]{IEEEtran}

\input{input.tex}

\title{Optimal WiFi Sensing via Dynamic Programming}

\author{Abhinav Kumar, Rahul Vaze, Sibi Raj B Pillai, Aditya Gopalan}

\begin{document}
\maketitle
\begin{abstract} The problem of finding an optimal sensing schedule for a mobile device that encounters an intermittent WiFi access opportunity is considered. 
At any given time, the WiFi is in any of the two modes, ON or OFF, and the mobile's incentive is to connect to the WiFi in the ON mode as soon as possible,
 while spending as little sensing energy. We introduce a dynamic programming framework which enables the characterization of an explicit solution for
several models, particularly when the OFF periods are exponentially distributed.   

While the problem for non-exponential OFF periods is ill-posed in general, a usual workaround in literature is to make the mobile device aware
if one ON period is completely missed. In this restricted setting, using the DP framework, the deterministic nature of the optimal sensing policy is established, and value iterations are shown to converge to the optimal solution.
Finally, we address the blind situation where the distributions of ON and OFF periods are unknown. A continuous bandit based learning algorithm 
that has vanishing regret (loss compared to the optimal strategy with the knowledge of distributions) is presented,  and comparisons with 
the optimal schemes are provided for exponential ON and OFF times.
\end{abstract}

\section{Introduction}
The available WiFi connectivity in mobile environments can be intermittent. 
In an effort to maximize WiFi connectivity time, current smartphones keep scanning/sensing for WiFi connection quite frequently, however, they loose precious battery life in this process.  The sensing schedule clearly depends on the distributions of the ON and OFF periods of the WiFi APs. This paper is an effort in finding the optimal sensing periods  given the knowledge of the ON and OFF period distributions. 

Given a geographical area with a fixed number of WiFi APs and a roaming mobile, the WiFi connection opportunity can be modeled as a two-state Markov chain with $\{\text{ON, OFF}\}$ states. In \cite{kim2011improving}, it is shown that the ON and OFF periods can be well approximated by exponential distributions. Without explicitly counting for the sensing cost, \cite{kim2011improving} also found the optimal sensing durations that minimize the  rate of missed ON periods. The analysis, however, is not completely rigorous, for example, the missed ON period in a given time period does not depend on the length of the time period, which is anomalous.

A natural metric for finding the optimal sensing duration is the sum of the expected length of the missed ON periods and the expected sensing cost \cite{Yung2013}. Even though \cite{Yung2013} considered this metric, however, for analysis, the metric was simplified, for example by replacing some of the random variables with their expectations.
Optimal solutions to these approximations for general ON and OFF distributions were presented in \cite{Yung2013}.
%
Some heuristic solutions \cite{kim2011improving, wang2009opportunistic} have also been found that modulate the sensing durations given the frequency of failure of detection. Some other practical smart sensing protocols for WiFi sensing can be found in \cite{zhou2010zifi,wu2009footprint}. Sensing in cognitive radio is also similar to this work \cite{SwamiCR}, however there, 
the unlicensed users sense to maximize their throughput without harming the licensed users. The cognitive radio setting also leads to a partially observed Markov decision process.

A critical assumption in  \cite{Yung2013} is that the system is reset 
if one complete ON interval is lost/missed because of no sensing epoch lying in that ON period. This assumption is particularly required when the distribution of the OFF periods is not exponential, since  otherwise the problem becomes ill-posed. See remark \ref{rem:nonexp} for a detailed explanation. Under this assumption, the problem is restricted to one OFF and one ON period, 
where a policy schedules the channel senses till the first ON state is detected or missed.

In this paper, we consider the metric as the sum of the expected length of the missed ON periods and the expected sensing cost similar to \cite{Yung2013}. 
Unlike \cite{Yung2013}, \cite{kim2011improving}, our approach relies on a dynamic programming formulation.
We solve for the general problem when the OFF period is exponentially distributed, while the ON periods are IID with 
any arbitrary distribution. The DP framework also allows us to rectify the anamolies in the past work concerning exponential ONs and OFFs
\cite{kim2011improving}.
%
For the non exponentially distributed OFF periods, we consider the restriction of one OFF and ON period similar to \cite{Yung2013}, but do not change the metric to suit analysis as done in \cite{Yung2013}. Again posing the problem as a dynamic program,  we obtain  structural results that show that the optimal policy is deterministic, and which can be found via value iteration that is shown to converge to the optimal solution.  
The restricted problem can be seen as  a generalization of \cite{azad2011optimal}, where the ON period never expires.

Almost all prior work on smart WiFi sensing assumes the knowledge of the distribution of the OFF and ON period distributions. In practice, that can be obtained only via training, however, is costly in terms of resources. To overcome this, we propose a blind learning framework, where the learning algorithm learns the optimal sensing duration iteratively, without any training.
The proposed algorithm is inspired by algorithms for continuous bandit problems \cite{kleinberg2004nearly, auer2007improved}, where each agent has a continuum of strategies to choose and its objective is to maximize a reward function, however, it does not know the reward distribution conditioned on its choice. We show that the proposed algorithm (following \cite {auer2007improved}) for finding the optimum sensing durations has a vanishing regret as a function of time, where regret is defined as the difference between the reward of an optimal algorithm with the knowledge of the distribution, and the blind learning algorithm. For lack of space, we illustrate the vanishing regret only when the underlying OFF-ON period distribution is exponential, but it easily applies for any other distribution.

\section{System Model}
\label{sec:sysmod}

Consider a mobile device that is moving in and out of WiFi APs' transmission radii, and encounters intermittent WiFi access opportunities in time, as shown in Fig. \ref{fig:model:1}. We assume that at time $t$, AP state is OFF 
if the mobile device is not in any AP's transmission radius, and ON otherwise. Thus, as shown in Fig. \ref{fig:model:1}, the mobile sees alternating ON and OFF periods, where it can receive data only in the ON periods. 

  \begin{figure}[h]
\centering
\includegraphics[width=0.44\textwidth]{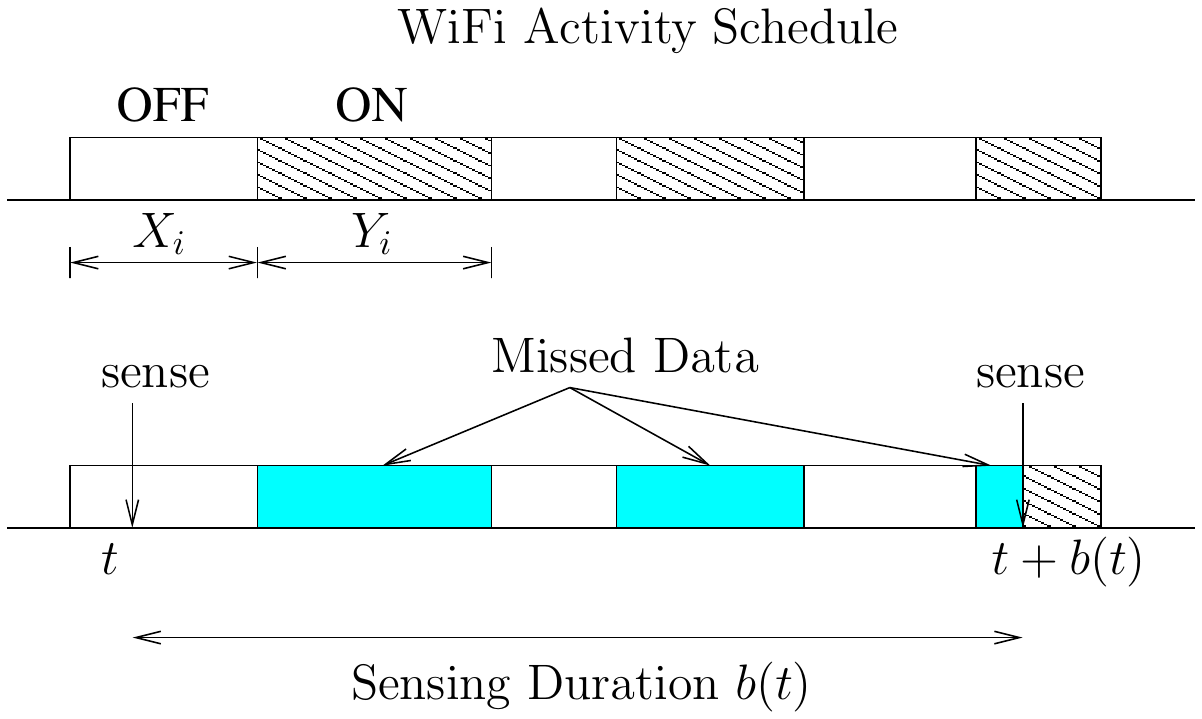}
\caption{System model description \label{fig:model:1}}
\end{figure}
To detect  ON periods, the mobile device employs sensing. If on sensing at  time $t$, the AP state is found ON, the device gets connected to the 
AP till the end of that ON period. We assume that the device learns about the disconnection as soon as the ON period is over, by using either the rapid increase in error probability or no useful data transmission. 
Otherwise, if on sensing at time $t$, the AP state is OFF, then the mobile decides to sleep and decides on the duration of the next sensing epoch $t+b(t)$, as shown in Fig. \ref{fig:model:1}. 

To save on energy, mobile senses intermittently, and consequently loses out on connecting to the AP as soon as the ON period starts. In particular, the shaded region (missed data) in Fig. \ref{fig:model:1} represents the lost opportunity because of intermittent sensing. Longer sleep periods incur less sensing power consumption
but decrease the WiFi connectivity time utilized, while shorter sleep periods increase the WiFi connectivity time at the cost of increasing the sensing power consumption.
 
To strike a balance between the lost ON period time and the sensing power consumption, we consider the problem of
finding the sensing intervals so as to minimize the sum of expected lost opportunity for data reception and the expected power for sensing. Now we make this formal. 

Let the duration of the $i^{th}, i\ge 1$ OFF and ON period be denoted by $X_i$ and $Y_i$, respectively, as shown in Fig. \ref{fig:model:1}. We assume that both $X_i$ and $Y_i$ are independent for $i\ge 1$. The PDF of $X$ and $Y$ is denoted by $f_d(x)$ and $f_c(y)$, where the subscript $d$ and $c$ represent disconnection and connection, respectively. 

If a sensing reveals the AP state to be ON, there is no decision to make, and the mobile device stays connected from there on till the end of the current ON period, and get disconnected at the end of it, and the system restarts.  
The non-trivial decision problem is when the current sense reveals the AP state to be OFF. 
We define an ON period to be a {\it discovered} ON period, if a sensing epoch lies in that ON period. In a discovered ON period, {\it useful} ON time is the time between the sensing epoch and the end of the discovered ON period. An illustration is provided in Fig. \ref{fig:learning:2}.
Time period between the end of two consecutive discovered ON periods is defined to be a {\it session}. 
Recall that system resets at the the end of each discovered ON period, thus we focus on any one particular session here onwards.

  \begin{figure}[h]
\centering
\includegraphics[width=0.44\textwidth]{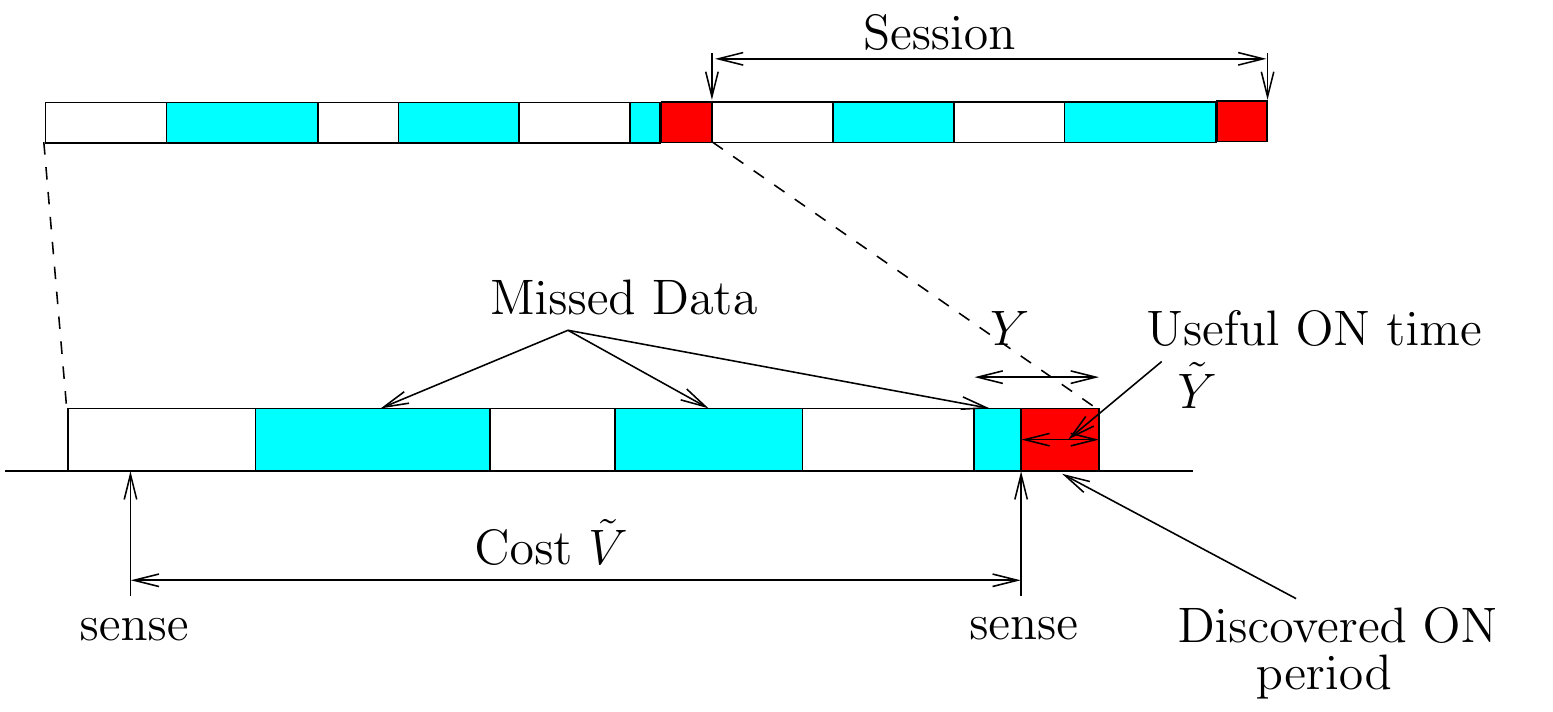} 
\caption{Illustration of sessions for learning algorithm \label{fig:learning:2}}
\end{figure}

Let $\1_{\mathsf{on}}(t)$ ($\1_{\mathsf{off}}(t)$) denote the event that the AP is in ON (OFF) state at time $t$.  
$P_{\mathsf{off}}(t)$ be the probability that the AP is in OFF state at time $t$, and $P_{\mathsf{off}}(t+x| t\in \mathsf{*})$, be the probability that the AP is in OFF state at time $t+x$ given that AP is in state $*, *\in \{\text{OFF, ON}\}$ state at time $t$. 
We define the cost between two sensing epochs at $t$ and $t+b(t)$ in a session as 
\begin{equation}\label{eq:running cost}
c(t,b(t)) = c_s + M(t,b(t)),
\end{equation}
where $c_s$ is fixed sensing cost, and $M(t,b(t))$ is the missed/lost ON time between time $t$ and $t+b(t)$.
Then, the sensing problem can be cast as a dynamic problem (DP), 
\begin{equation}\label{eq:dp}
V(t) = \max_{b(t)\ge 0} [\bbE\{c(t,b(t))\} + P_{\mathsf{off}}(t+b(t)| t\in \mathsf{off}) V(t+b(t))],
\end{equation}
where we have assumed that if at time $t$ $b(t)$ is selected as the next sensing duration, the running cost is $\bbE\{c(t,b(t))\}$, and the process restarts if AP is in OFF state at time $t+b(t)$, which happens with probability $P_{\mathsf{off}}(t+b(t)| t\in \mathsf{off})$. 

Given that an OFF period is going on at time $t$, we define the residual OFF time (time of the completion of OFF period) at time $t$ as $r_t$, that has CDF  $F_{r_t}(x) = P(X>x| X>t)$, where $X$ represents the duration of the OFF period. Note that $P_{\mathsf{off}}(t+b(t)| t\in \mathsf{off})=r_t(b(t))$ in \eqref{eq:dp}.

\begin{remark}\label{rem:nonexp} If the distribution of the OFF periods is not exponential, then \eqref{eq:dp} is not well-defined. To see this, consider that if for two consecutive sensing times $t$ and $t+b(t)$, the AP state is OFF, the distribution of the residual OFF time starting from $t+b(t)$ $(r_{t+b(t)})$ is not well-defined since we do not know when the current OFF period started. 
To handle the non-exponential distribution of the OFF periods, we will follow the approach of \cite{Yung2013} in Section \ref{sec:rest}, where it is assumed that as soon as any one complete ON interval is lost/missed because of no sensing epoch lying in that ON period, the system is reset.
\end{remark}

In light of Remark \ref{rem:nonexp}, in this section, we restrict our attention to exponential distribution for the OFF period, while the ON period is allowed to have any arbitrary distribution. Under this assumption, we have that the optimal control $b(t)$ does not depend on $t$. 
\begin{lemma} The optimal control $b(t)$ that solves \eqref{eq:dp} does not depend on $t$, when OFF periods are exponentially distributed.
\end{lemma}
\begin{proof} At time $t$, the next sensing duration $b(t)$ is decided only if $t \in \mathsf{off}$. However, because of the memoryless property of the OFF periods, the event that $t \in \mathsf{off}$ gives no information about the future length of OFF and ON periods, the optimal sensing duration $b(t)$ does not depend on $t$.
\end{proof}

With arbitrary ON period distribution, we need the following notation. Let $P_{\mathsf{off}}(y | z \uparrow)$ be the probability that the AP is in OFF state at time $y, y\ge z$ given that a transition from OFF to ON period happens at time $z$. 
Because of memoryless property of the OFF period distribution,we do not need such notation for transition from ON to OFF period, since $P_{\mathsf{off}}(y | z \downarrow) = P_{\mathsf{off}}(y | z \in \mathsf{off})$.
To further the analysis, we next find an expression for $P_{\mathsf{off}}(t)$.
\begin{lemma}\label{lem:poff}
\begin{multline}\label{eq:poff1}
P_{\mathsf{off}}(t+x| t\in \mathsf{off}) =  \\ 
 P(r_t \geq x) + \int_0^x f_d(z) P_{\mathsf{off}}(t+x | z \uparrow)dz, z\ge t
 \end{multline} 
 \vspace{-.3in}
 \begin{multline}\label{eq:poff2}
P_{\mathsf{off}}(t+x | t \uparrow) = \int_0^x f_c(w) P_{\mathsf{off}}(t+x | w\in \mathsf{off}) dw, w\ge t.
\end{multline}
\end{lemma}
\begin{proof} The first expression is obtained by counting the two exclusive events, i) the residual OFF time $r_t$ of the present OFF period that is going on at time $t$ exceeds $x$, and ii) the present OFF period expires at time $z$ (i.e. OFF to ON transition happens at $z$), 
and taking the expectation of $P_{\mathsf{off}}(t+x | z \uparrow)$ with respect to $t\le z \le x$. The second expression follows similarly. 
\end{proof}

\begin{corollary}\label{prop:poff} For OFF period  $\sim EXP(\lambda_d)$ and ON period  $\sim EXP(\lambda_c)$, for $x \geq 0$,
$$P_{\mathsf{off}}(t+x| t\in \mathsf{off}) =  \frac{\lambda_c}{\lambda_c + \lambda_d} + \frac{\lambda_d}{\lambda_c + \lambda_d} \exp(-(\lambda_d + \lambda_c)x).$$
\end{corollary}
Note that $P_{\mathsf{off}}(t+x| t\in \mathsf{off})$ does not depend on the starting time $t$ as expected, because of the memoryless property of the exponential distribution.
\begin{proof}
We use the Laplace transforms of \eqref{eq:poff1} and \eqref{eq:poff2} to solve for 
$P_{\mathsf{off}}(t+x| t\in \mathsf{off})$. With OFF period  $\sim EXP(\lambda_d)$ and ON period  $\sim EXP(\lambda_c)$, $f_d(w) = \lambda_d \exp(-\lambda_d w)$ and $f_c(w) = \lambda_c \exp(-\lambda_c w)$. Denoting the Laplace transform of $f_d(w)$ as $f_d^*(s)$,  $f_c(w)$ as $f_c^*(s)$, $P_{\mathsf{off}}(t+x| t\in \mathsf{off})$ as $\hat P_0^{of}(s)$ and 
  $P_{\mathsf{off}}(t+x | t \uparrow)$ as $P_1^{of}(s)$, we have from \eqref{eq:poff1} and \eqref{eq:poff2}, \begin{align*}
\hat P_0^{of}(s) &= \frac 1{s+\lambda_d} + f_d^*(s) \hat P_1^{of}(s), \\
\hat P_1^{of}(s) &= f_c^*(s) \hat P_0^{of}(s).
\end{align*}
Note that $f_c^*(s) = \frac{\lambda_c}{s + \lambda_c}$ and 
$f_d^*(s) = \frac{\lambda_d}{s + \lambda_d}$. Hence
\begin{align}
\hat P_0^{of}(s) &= \frac 1{s+\lambda_d} \frac 1{1 - f_d^*(s) f_c^*(s)}, \\
	&= \frac 1{\lambda_d + \lambda_c} \left( \frac{\lambda_c}{s} + \frac{\lambda_d}{s+\lambda_d + \lambda_c} \right).
\end{align}
Taking the inverse transform, for $x\ge 0$
\begin{align*}
P_{\mathsf{off}}(t+x| t\in \mathsf{off}) = \frac{\lambda_c}{\lambda_c + \lambda_d} + \frac{\lambda_d}{\lambda_c + \lambda_d} \exp(-(\lambda_d + \lambda_c)x).
\end{align*}
\end{proof}

Next, we find the expected running cost $\bbE\{M(t,b(t))\}$ to compute the expected cost $\bbE\{c(t,b(t)\}$. Again appealing to the memoryless property of the exponential distribution,  $\bbE\{M(t,b(t))\} = \bbE\{M(b)\}$ where we have shifted the starting time to $0$.
We will use recursions similar to \eqref{eq:poff1} and \eqref{eq:poff2} and Laplace transforms to find $\bbE\{M(t)\}$.

Let $\sfM_{d}(t) = \bbE\{M(t)\}$ be the average missed ON period time between times $\tau$ to $t+\tau$, when $\tau \in \mathsf{off}$. Moreover, let 
$\sfM_{\uparrow}(t) = \bbE\{M(\tau, \tau+t)\}$ average missed ON period time between times $\tau$ to $t+\tau$ given that 
the OFF to ON transition happens at time $\tau$, and similarly let $\sfM_{\downarrow}(t) = \bbE\{M(\tau, \tau+t)\}$ given that 
the ON to OFF transition happens at time $\tau$, where the LHS has no dependence on $\tau$ because of the Markov property of ON and OFF periods. So without loss of generality, we take $\tau=0$. Note that because of memoryless property of OFF times $\sfM_{d}(t) = \sfM_{\downarrow}(t)$.

\begin{lemma}
 \begin{equation} \label{eq:tmc1}
 \sfM_{\uparrow}(t)=t\int_t^\infty f_c(x) dx+\int_0^tf_c(x)(x+\sfM_{\downarrow}(t-x))dx,
 \end{equation}
 and
 \begin{equation}\label{eq:tmc2}
 \sfM_{\downarrow}(t)=\int_0^tf_d(x)\sfM_{\uparrow}(t-x)dx.
 \end{equation}
 \end{lemma}
 \begin{proof}
To derive \eqref{eq:tmc1}, we have broken the expectation $\sfM_{\uparrow}(t)$ into two terms, where in the first we count the expected length of the ON period that starts at time $\tau = 0$ and continues beyond time $t$, and in the second, we consider the case when the ON period that starts at time $\tau = 0$ finishes at some time $x < t$ and count for the expected loss with ON to OFF transition happening at $x$. The second expression \eqref{eq:tmc2} follows similarly.
\end{proof}

\begin{theorem}\label{thm:lt}
For OFF period  $\sim EXP(\lambda_d)$ and ON period  $\sim EXP(\lambda_c)$, the expected loss $\sfM_{d}(t)$ is given by
\begin{equation}
\sfM_{d}(t) = \frac{\lambda_d}{\lambda_d+\lambda_c}\left( t - \frac{1-e^{-(\lambda_d+\lambda_c)t}}{\lambda_c + \lambda_d}\right).
\end{equation}
\end{theorem}
\begin{proof}
 We take the Laplace transforms of \eqref{eq:tmc1} and \eqref{eq:tmc2} to get 
\begin{equation*}
\sfM_{\uparrow}^*(s)=\frac{(f_c^*(0)-f_c^*(s))
}{s^2}+f_c^*(s)\sfM_{\downarrow}^*(s),
\end{equation*}
and
$\sfM_{\downarrow}^*(s)=f_d^*(s)\sfM_{\uparrow}^*(s)$.
So we have 
\begin{equation*}
\sfM_{d}^*(s) = \sfM_{\downarrow}^*(s)=f_d^*(s)\frac{(f_c^*(0)-f_c^*(s))}{s^2(1-f_c^*(s)f_d^*(s))}.
\end{equation*}
Substituting for $f_c^*(s) = \frac{\lambda_c}{s+ \lambda_c}$ and 
$f_d^*(s) = \frac{\lambda_d}{s + \lambda_d}$, and taking the inverse Laplace transform we obtain the result.
\end{proof}

\begin{remark} It is important to note that similar derivation for $\sfM_{d}(t)$ has been attempted in \cite{kim2011improving}, however, there are glaring errors in it. For example, the ON period loss $\sfM_{d}(t)$ incurred in time $t$ does not depend on $t$, and is always less than $1$. 
\end{remark}

Finally, we have all the intermediate results to solve for the DP \eqref{eq:dp} when the OFF periods are exponentially distributed, where $b(t) =b$, and the DP simplifies to 
\begin{equation}\label{eq:dpsimple}
V(b) = \max_{b\ge 0} [\bbE\{c(b)\} + P_{\mathsf{off}}(b) V(b)].
\end{equation}
\begin{theorem}\label{thm:b}
The optimal sensing duration $b$ satisfies the following equation 
\begin{equation}
\frac{\mathrm{d}}{\mathrm{d}b} V = \frac{\mathrm{d}}{\mathrm{d}b} \left(\frac{c_s + \sfM_{d}(b)}{1-P_{\mathsf{off}}(b)}\right)=0,
\end{equation}
where $\sfM_{\downarrow}(b)= \sfM_d(b)$ can be found by substituting for $f_c(x)$ and $f_d(x)$ in  \eqref{eq:tmc1} and \eqref{eq:tmc2} and $P_{\mathsf{off}}(b)$ can be found by substituting for $f_c(x)$ and $f_d(x)$ in \eqref{eq:poff1} and \eqref{eq:poff2}.
\end{theorem}
\begin{proof} Follows by rewriting \eqref{eq:dpsimple}, and taking the $V$ terms common, and equating the derivative of $V$ with respect to $b$ to zero. 
\end{proof}

\begin{corollary}\label{cor:optbexpexp}
For OFF period  $\sim EXP(\lambda_d)$ and ON period  $\sim EXP(\lambda_c)$, the optimal sensing duration $b$ satisfies
\begin{align} \label{eq:lambert:1}
e^{-(\lambda_c+\lambda_d)b} \left( 1 + \frac {c_s}{\lambda_d} (\lambda_c+\lambda_d)^2 + b (\lambda_c + \lambda_d)\right) = 1.
\end{align} 
\end{corollary}
\begin{proof} From Corollary \ref{prop:poff} and Theorem \ref{thm:lt}, substituting for $P_{\mathsf{off}}(b)$ and $\sfM_{\uparrow}(b)$ in \eqref{eq:dpsimple},  we get
\begin{align}\label{eq:VEXP}
V(b)  &= 
	\frac{ c_s + \frac{\lambda_d}{\lambda_d+\lambda_c}\left( b - \frac{1-e^{-(\lambda_d+\lambda_c)b}}{\lambda_c + \lambda_d}\right) }
	{\frac{\lambda_d}{\lambda_c + \lambda_d} \left( 1 - e^{-(\lambda_d + \lambda_c)b}  \right)}.
\end{align}
Equating $\frac{\mathrm{d}}{\mathrm{db}} V = 0$, we get~\eqref{eq:lambert:1}. 
While this is a trascendental equation, the numerical solution is easy, see Figure~\ref{fig:sim:1}.
It is also easy to check that the second derivative of $V$ is $\ge 0$ and hence the above solution is indeed the global minimum. 
\end{proof}

Similarly, we can find the optimal sensing duration $b$ for other ON period distributions as long as the OFF period distribution is exponential.

\section{Non EXP-OFF Period}\label{sec:rest}
Recall from Remark \ref{rem:nonexp} that in the framework of Section \ref{sec:sysmod}, we cannot solve for the optimal sensing durations when the OFF period distribution is not exponential. To circumvent this restriction, in this section we make an extra assumption following \cite{Yung2013}, where if any complete ON period is missed because of no sensing in that ON period, the mobile device is made aware of that and the system is reset. Thus, the problem \eqref{eq:dp} is now restricted to one OFF ($X$) and one ON ($Y$) period, and we want to choose the sensing durations so as to minimize the sum of the expected missed ON period time and the expected sensing cost. 

Thus, in this case, the DP \eqref{eq:dp} is,  
\begin{equation}\label{eq:dprest}
V(t) = \max_{b(t)\ge 0} [\bbE\{c(t,b(t))\} + P(r_t>b(t)) V(t+b(t))],
\end{equation}
where the cost function $c(t,b(t))$ is simplified and given by $c(t,b(t)) = c_s + \bbE\{(b(t) - r_t)\1_{b(t) \ge r_t, Y> b(t)-r_t}\} + \bbE\{ Y \1_{(Y \le b(t)- r_t)}\}$, where $c_s$ is the fixed sensing cost, while the lost ON time is written as two terms,  either the complete one ON period $Y$ is missed if $b(t)$ is larger than $r_t + Y$, otherwise, the missed ON time is $(b(t) - r_t)$. 

Problem \eqref{eq:dprest} is a generalization of problem considered in \cite{azad2011optimal}, where the length of the ON period is infinite (not a random variable) and the problem is to minimize the sum of the expected time lost in detecting ON period and the expected sensing cost. Note that in our setup, since the ON period expires in finite expected time, solution of \cite{azad2011optimal} does not apply.

We use a state space approach to derive results when the ON and OFF periods  have a general distribution. The state space we consider is the set of non-negative real numbers. An action $b(t)$ is the duration of the next sleep period. We assume that $b(t)$ can take values only in a finite set (which is clearly true in practice). Thus, the set
of $t$ reachable (with positive probability) by any policy is
countable and hence without loss of generality we assume that
the state space is discrete. Then we have the following result.

\begin{theorem}\label{thm:struct}
To solve \eqref{eq:dprest}, for any ON and OFF period distribution, the following statements hold.
\begin{enumerate}
\item There exists an optimal deterministic stationary policy.
\newline
\item Let $V^0 = 0$, $V^{k+1}=\cL V^{k}$, where
     $\cL V(t)=\min_{b(t)} [c(t,b(t))+P(r_t>b(t))V(t+b(t)) ]$
     and $c(t,b(t))$ is the per stage/running cost. Then $V^k$ converges monotonically to the optimal value $V^*$.
\newline
\item $V^*$  is the smallest nonnegative solution of $V^*  = \cL V^*$.
A stationary policy that chooses at state (time) $t$ an action that
achieves the minimum of $\cL V^*$  is optimal.
\end{enumerate}
\end{theorem}
\begin{proof}
1) follows from the [\cite{BookPuterman}, Thm 7.3.6] that states that if the state space is discrete (finite or countable) and the action set for any state is finite, then there exists an optimal deterministic stationary policy. These conditions are satisfied in this case since 
the set of all actions ($b$'s) is assumed to be finite, and the state space is countable. 
Similarly, 2) follows from the [\cite{BookPuterman},
Thm 7.3.10] that states that if reward $w$ for action $a$ at state $s$, $w(s,a)\ge 0$, and state space is countable and action space is finite for each state, then if $V^0=0$, $V^{n+1}=\cL V^n$ converges monotonically to $V^*$ and 3) follows from [\cite{BookPuterman}, Thm 7.3.3].\end{proof}
Theorem \ref{thm:struct} shows that it is sufficient to consider deterministic policies without losing out on optimality, and randomized strategies are not needed. Moreover, part $2)$ and $3)$ tell us that the value iteration policy converges to the optimal solution for any ON and OFF period distributions.

We now consider the special case when the OFF period depends on the time at which it starts, but in the limit of very large time $t$, it loses that dependence.

\begin{theorem}\label{thm:convrestime}
Assume that the residual OFF time $r_t$ converges in distribution to $\sfr$, and  
define $v(b)=\frac{ c^{\star}(b)}{1-P(\sfr > b)}$. Then 
\begin{enumerate}
\item $\lim_{t\to \infty}V^*(t)=\min_{b} v(b)$.
\newline
\item Assume that there is a unique $b$ that achieves the minimum of $v(b)$ and denote it by $b^{\star}$. Then there is some stationary optimal policy $b(t)$ such that
for all $t$ large enough, $b(t) =b^{\star}$.
\end{enumerate}
\end{theorem}
\begin{proof}
Let $V^0=0$, and assume that $\bar V^k = \lim_{t \to \infty}V^k(t)$ exists for some $k$. Then from definitions used in Theorem \ref{thm:struct}, we have
\begin{align*}
\bar V^{k+1}&=\lim_{t \to \infty}\cL V^k(t),\\
&=\lim_{t \to \infty}\min_{b(t)}[c(t,b(t))+P(r_t>b(t))V^k(t+b)],\\
&=\min_{b}[ c^{\star}(b) +P(\sfr >b)\bar V^k],
\end{align*}
where the last equality follows since $r_t$ converges in distribution to $\sfr$, and from the bounded convergence theorem 
$$\lim_{t\to \infty}c(t,b(t)) \to  c^{\star}(b).$$
Essentially, since $r_t$ converges in distribution to $\sfr$, the per-stage/running cost $c(t,b(t))$ becomes independent of $t$ as $t\rightarrow \infty$ (similar to the case when OFF periods are exponentially distributed).
Hence by convergence of $V^{k}$ to $V^*$ by Theorem \ref{thm:struct}, the limit $\bar V = \lim_{t \to \infty}V^*(t)$ exists. Thus, there exists a constant deterministic policy as $t \to \infty$ which 
we denote by $b^{\star}$. This gives $\bar V= c^{\star}(b)+P(\sfr > b)\bar V$ which on rearranging gives (i). (ii) can be obtained by noting that $b^{\star}$ performs better than any other policy so the optimal solution $b(t)$ must tend to $b^{\star}$ as $t \to \infty$
\end{proof}
Therefore, if the residual OFF period distribution converges in time, then the optimal sensing duration converges to a constant after sufficiently long time. 
\begin{remark} The condition in Theorem \ref{thm:convrestime} is trivially true for exponentially distributed OFF period. A more non-trivial example is when the OFF period has hyper-exponential distribution, for which the residual OFF time $r_t$ converges in distribution to some $\sfr$.
\end{remark}
\begin{remark} Theorem \ref{thm:struct} and Theorem \ref{thm:convrestime} are similar to Propositions III.2 and III.3 in \cite{azad2011optimal}.
\end{remark}
\subsection{Example}
Next, we consider an example where both the OFF and ON periods are uniformly distributed between $[0,L_f]$ and $[0,L_o]$, respectively. To use Theorem \ref{thm:struct}, with sensing duration $b(t)$, we write down the sensing cost $c(t,b(t))$ and the probability $P(r_t>b(t))$ that at the next sensing epoch we again encounter an OFF period. 

For OFF period distributed uniformly between $[0,L_f]$, the residual OFF time distribution 
\begin{equation}\label{eq:unifrt}
F_{r_t}(x) = P(r_t > x) = P(X > x | X>t) = \frac{L_f-x}{L_f-t},
\end{equation}
for $t \le x \le L_f$. Thus, $P(r_t>b(t)) = \frac{L_f-b(t)}{L_f-t}$, and to compute cost $c(t,b(t))$, we calculate 
$$\bbE\{(b(t) - r_t)\1_{b(t) \ge r_t, Y> b(t)-r_t}\} = \frac{\frac{L_o^2}{6} + (1-b(t))\frac{L_o}{2}}{L_f-t},$$
and
$$
\bbE\{ Y \1_{(Y \le b(t)- r_t)}\} = \frac{b(t)^2 - b(t)(L_f-t) + \frac{L_f^2+L_ft + t^2}{3}}{2L_o},$$ where the total cost $c(t,b(t)) = c_s + \bbE\{(b(t) - r_t)\1_{b(t) \ge r_t, Y> b(t)-r_t}\} + \bbE\{ Y \1_{(Y \le b(t)- r_t)}\}$.
 
Hence to solve for the optimal sensing durations $b(t)$ via the value iteration method, we start with $V=0$, and write 
$V^{k+1}=\cL V^{k}$, where
  $$\cL V(t)=\min_{b(t)} [c(t,b(t))+P(r_t>b)V(t+b(t)) ],$$ where we substitute for $c(t,b(t))$ from above. From Theorem \ref{thm:struct}, these iterations converge to the optimal policy.
\section{Learning Framework}
In Sections \ref{sec:sysmod} and  \ref{sec:rest}, we have derived the optimal sensing duration assuming the knowledge of the distribution of the OFF and ON periods. In practice, learning these distributions is a problem in its own right. To obviate the need for exactly learning the distribution (might take a long training time), in this section, under the general model of Section \ref{sec:sysmod}, we present a continuous-armed bandit problem type formulation \cite{kleinberg2004nearly,auer2007improved}, where the algorithm learns the best sensing duration without explicitly knowing the underlying OFF and ON period distribution. For ease of exposition, we will assume that the OFF and ON periods are exponentially distributed with unknown parameters. The analysis carries over to all distributions for which the cost $V$ has continuous second derivatives.

An online learning  algorithm chooses one possible sensing duration $b$ for each session (defined earlier), and receives a reward that counts for the useful ON time and the cost incurred (lost ON period time and sensing cost). Depending on past choices of $b$, the algorithm modulates its choice of $b$ in future sessions in pursuit of larger rewards.
  \begin{figure}[h]
\centering
\includegraphics[width=0.44\textwidth]{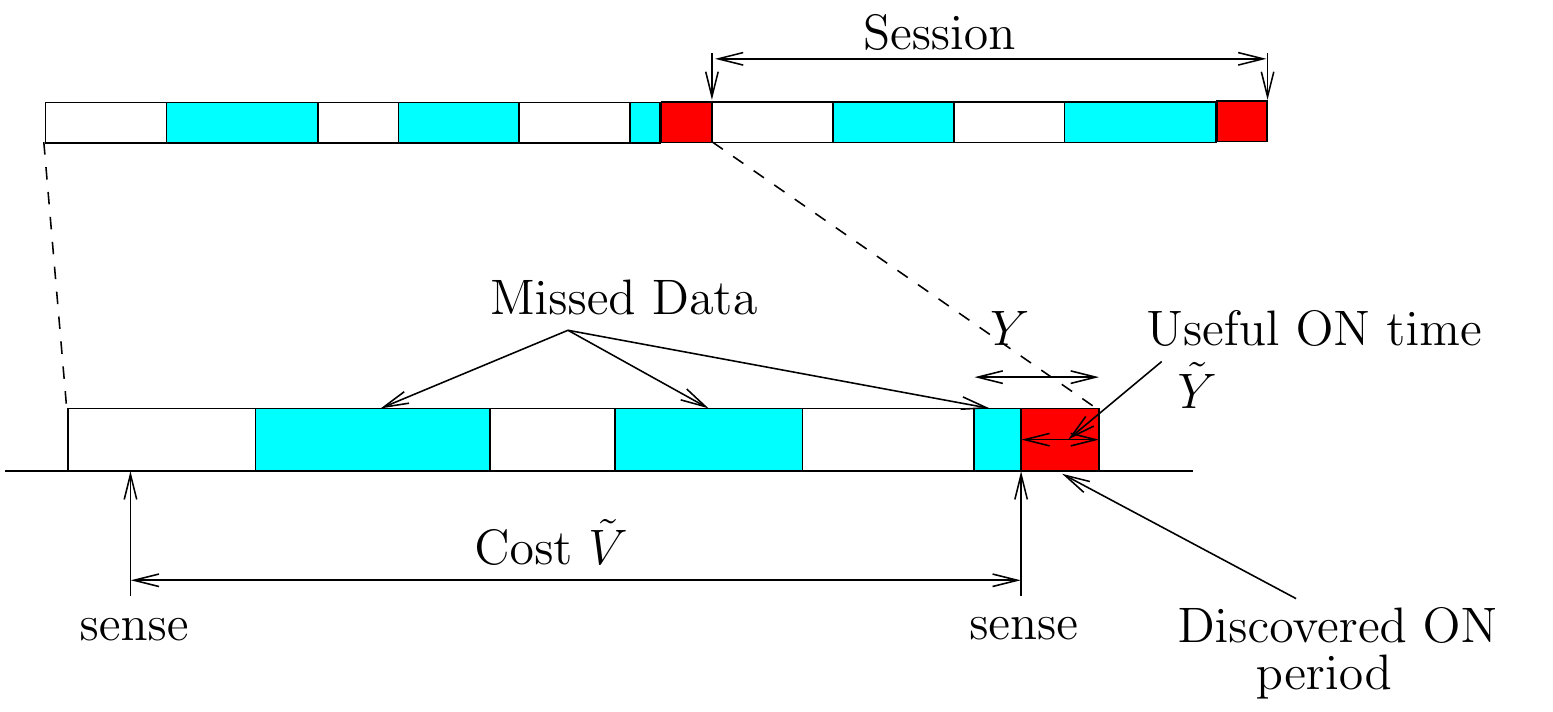}
\caption{Illustration of sessions for learning algorithm \label{fig:learning} }
\end{figure}

The reward in session $i$ is $U_i = {\tilde Y} - {\tilde V}$, where as shown in Fig. \ref{fig:learning}, ${\tilde Y}$ is the length of the useful ON period (or discovered ON time), and ${\tilde V}$ is the random variable whose expectation is cost \eqref{eq:VEXP}, that counts the sensing cost and missed ON periods in each session. Note that there could be multiple sensing epochs in each session, and the sensing cost of each session is $c_s$ times the number of sensing epochs at which an OFF period is sensed in that session. 
Let OFF period  $\sim EXP(\lambda_d)$ and ON period $\sim EXP(\lambda_c)$, with unknown parameters 
$\lambda_d$ and $\lambda_c$, where we assume that $\lambda_d$ and $\lambda_c$ are such that the optimal $b^{\star} \in [0, b_{max}]$ from Corollary \ref{cor:optbexpexp}.

The online algorithm's objective is to minimize the expected regret, $$\min_{b(i), i=1,\dots,T} \bbE\{R(T)\},$$ by choosing action $b(i)$ is session $i$,
and
\begin{equation}\label{eq:regret}
R(T) = \sum_{i=1}^T U_i^{\star} - \sum_{i=1}^T U_i,
\end{equation}
where $U_i^{\star}$ is the optimal reward knowing $\lambda_d$ and $\lambda_c$, i.e. playing optimal $b$ from Corollary \ref{cor:optbexpexp} in each session, and $T$ is the time horizon. 

The learning algorithm called the {\bf OnlineLearning} \cite{auer2007improved} to find the sensing duration $b$ to minimize the expected regret is given at the top of the page.
\begin{table}
\begin{tabular}{r l}
\hline
& \textbf{{\bf OnlineLearning}}\\
\hline
1	&	\text{Choose $n$} \\
2	&	\text{Divide $[0, b_{max}]$ into $n$ intervals $I_k =b_{max}[\frac{k-1}{n}, \frac{k}{n}]$, $0\le k\le n$}\\
3	&	\text{For each $I_k$, choose a point (sensing duration $b$) uniformly at random}\\
4	&	\text{For $i=1:T$}\\
5	&	\text{Choose that interval $I_k$ that maximizes ${\hat U}_k + \sqrt\frac{2 \ln i}{t_k}$}, \\ 
& \text{where ${\hat U}_k$ is the average (empirical) reward obtained from points} \\  &\text{in interval $I_k$
so far, and $t_k$ is the number of times interval $I_k$ has} \\
&\text{been chosen till session $i$ and 
$i$ is the overall number of sessions so far}\\
6	&	\text{Choose a point uniformly at random from the chosen interval $I_k$.}\\
\hline
\end{tabular} 
\end{table}

\begin{lemma}(Theorem $1$ \cite{auer2007improved})\label{lem:learning} 
If the expected reward given a strategy has continuous second derivatives, and finite number of maximas, then the expected regret obtained by the {\bf OnlineLearning} algorithm is bounded as follows,
$$\bbE\{R(T)\} \le {\cal O}\sqrt{T\log T},$$ for $n=\left(\frac{T}{\log T}\right)^{1/4}$.
\end{lemma}  
Thus, the average regret $\left(\frac{\bbE\{R(T)\}}{T}\right)$ goes to zero with the {\bf OnlineLearning} algorithm even without knowing the underlying distributions.

\begin{theorem} Using the {\bf OnlineLearning} algorithm, the normalized regret $\frac{\min_{b(i), i=1,\dots,T} \bbE\{R(T)\}}{T}$   goes to zero with increasing number of sessions for finding the optimal sensing duration without the knowledge of OFF and ON period parameters $\lambda_d$ and $\lambda_c$.
\end{theorem}  
 \begin{proof} Note that the expectation $V$ \eqref{eq:VEXP} of the cost function ${\tilde V}$  has continuous second derivatives and finite number of maximas for a fixed strategy $b$, and $\bbE\{{\tilde Y}\}$ does not depend on $b$ because of memoryless property of exponential distribution. Thus, the expected reward $\bbE\{U_i\}$ given $b$ has continuous second derivatives and finite number of maximas, and the result follows from Lemma \ref{lem:learning}.
 \end{proof}
  
  In Fig. \ref{fig:sim:1}, we demonstrate the performance of the {\bf OnlineLearning} via simulation. We use $\lambda_d$ and $\lambda_c$ such that the expected OFF period length and ON period length is $3$ and $2$, respectively, 
 and plot the optimal sensing duration and sensing duration discovered by {\bf OnlineLearning} algorithm as function of the sensing cost $c_s$. We see that the {\bf OnlineLearning} algorithm closely tracks the theoretical optimum computed by Corollary~\ref{cor:optbexpexp}. Furthemore, the average cost incurred while connecting to the AP, closely matches for the two algorithms, as shown in Fig.~\ref{fig:sim:2}.

\begin{figure}
\includegraphics[scale=0.5]{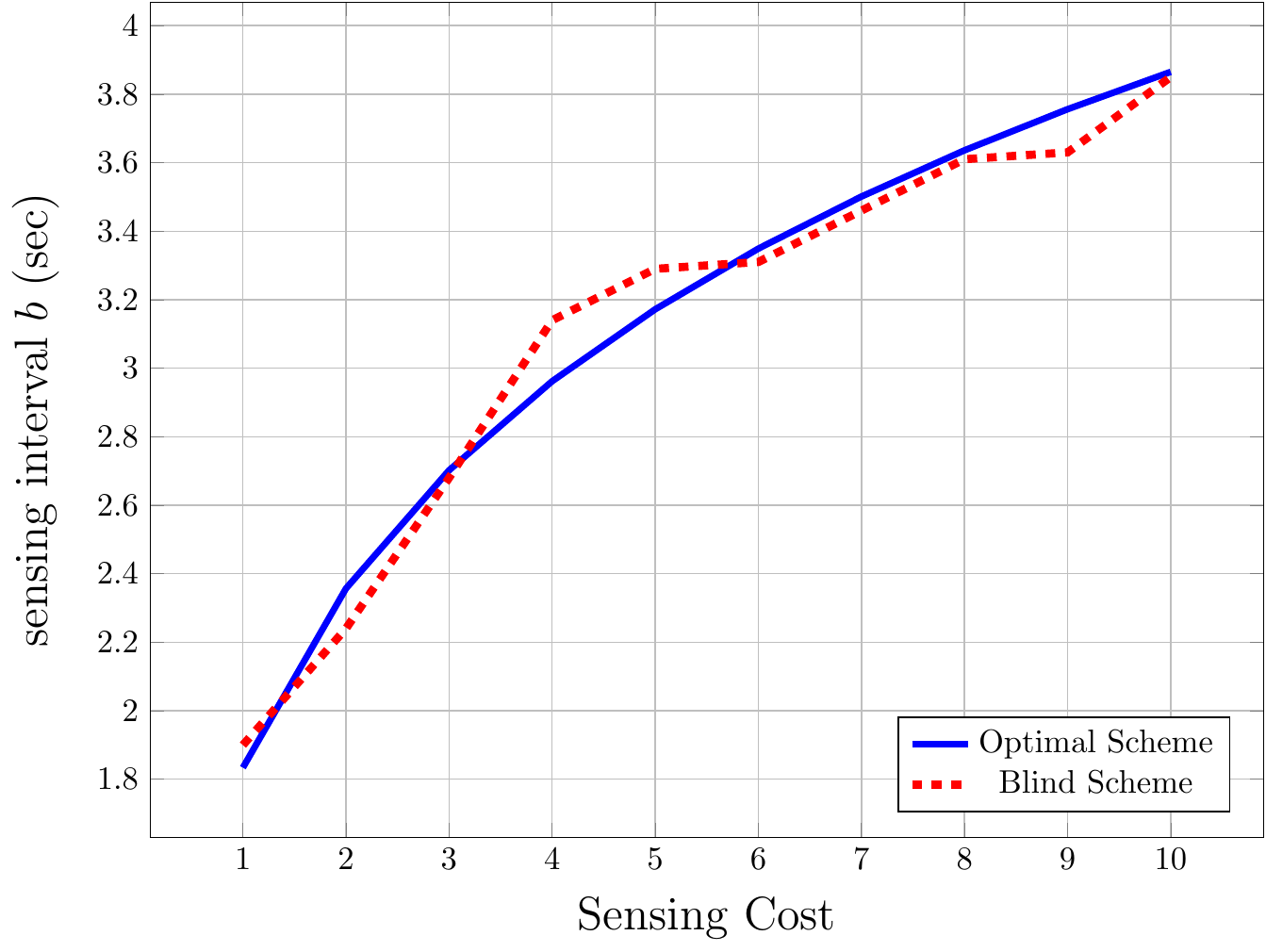}
\caption{Optimal sensing interval vs sensing cost \label{fig:sim:1}}
\end{figure}
  
\begin{figure}
\includegraphics[scale=0.5]{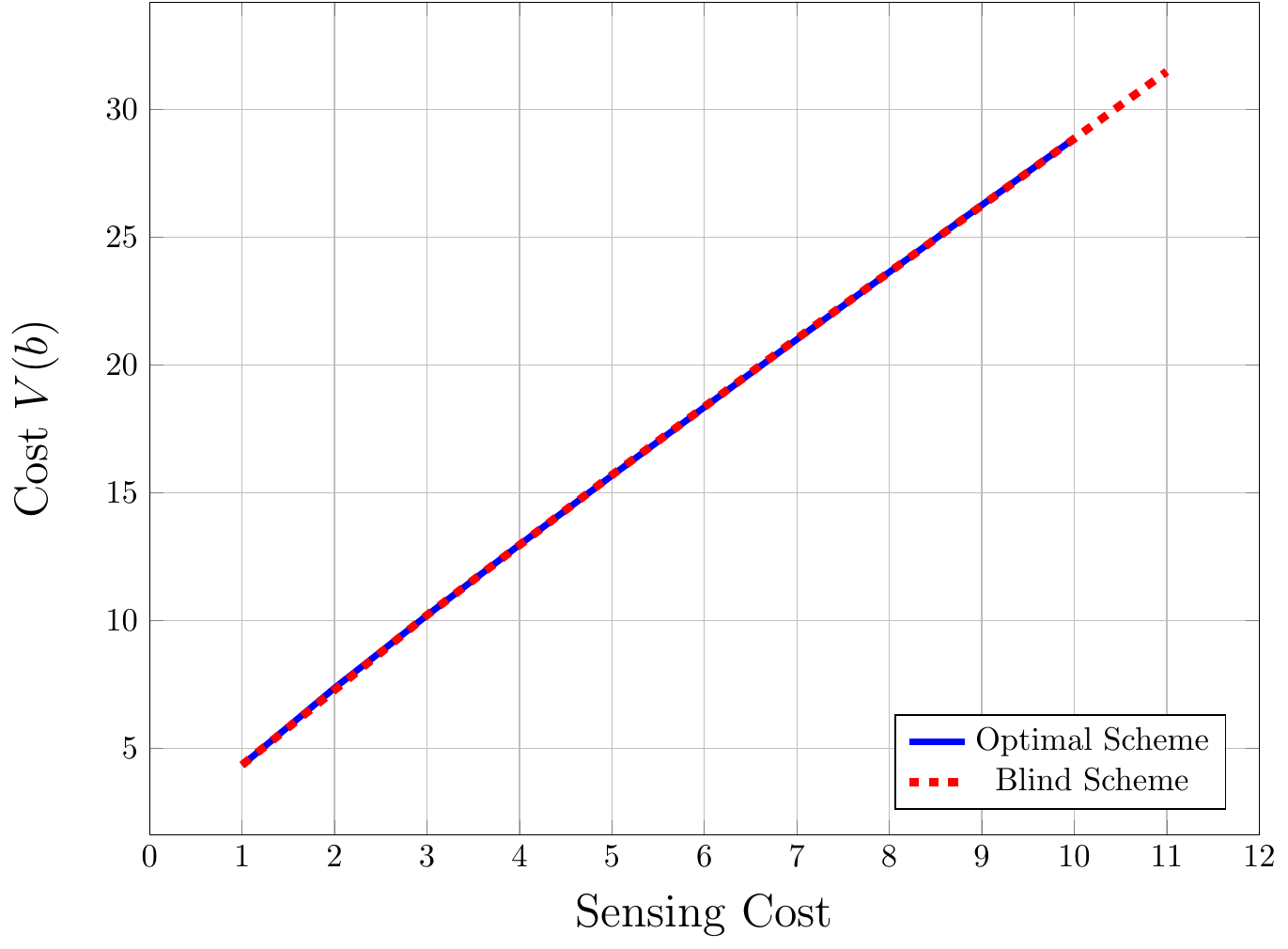}
\caption{Average  cost vs sensing cost\label{fig:sim:2}}
\end{figure}

\bibliographystyle{IEEEtran}
\bibliography{../IEEEabrv,../Research}

\end{document}

%% file: input.tex
\usepackage{amsfonts}
\usepackage{times}
\usepackage[pdftex]{graphicx}
\usepackage{latexsym}
\usepackage{dsfont}
\usepackage{amssymb}
\usepackage{amsmath}
\usepackage{cite}
\usepackage{verbatim}
\usepackage{subfigure}
\newtheorem{theorem}{Theorem}

\newtheorem{corollary}[theorem]{Corollary}

\newtheorem{lemma}[theorem]{Lemma}

\newtheorem{remark}[theorem]{Remark}
\newenvironment{proof}{ \textbf{Proof:} }{ \hfill $\Box$}

\def\bb0{{\mathbb{0}}}

\def\1{{\mathbf{1}}}

\def\bb{{\mathbf{b}}}

\def\b0{{\mathbf{0}}}

\def\bbE{{\mathbb{E}}}

\def\cL{\mathcal{L}}

\def\sfM{\mathsf{M}}

\def\b1{\mathsf{1}}

\def\sfr{{\mathsf{r}}}

\def\sf0{{\mathsf{0}}}









%% file: ICCPaperFinal.bbl
\begin{thebibliography}{10}
\providecommand{\url}[1]{#1}
\csname url@samestyle\endcsname
\providecommand{\newblock}{\relax}
\providecommand{\bibinfo}[2]{#2}
\providecommand{\BIBentrySTDinterwordspacing}{\spaceskip=0pt\relax}
\providecommand{\BIBentryALTinterwordstretchfactor}{4}
\providecommand{\BIBentryALTinterwordspacing}{\spaceskip=\fontdimen2\font plus
\BIBentryALTinterwordstretchfactor\fontdimen3\font minus
  \fontdimen4\font\relax}
\providecommand{\BIBforeignlanguage}[2]{{%
\expandafter\ifx\csname l@#1\endcsname\relax
\typeout{** WARNING: IEEEtran.bst: No hyphenation pattern has been}%
\typeout{** loaded for the language `#1'. Using the pattern for}%
\typeout{** the default language instead.}%
\else
\language=\csname l@#1\endcsname
\fi
#2}}
\providecommand{\BIBdecl}{\relax}
\BIBdecl

\bibitem{kim2011improving}
K.-H. Kim, A.~W. Min, D.~Gupta, P.~Mohapatra, and J.~P. Singh, ``Improving
  energy efficiency of {Wi-Fi} sensing on smartphones,'' in \emph{INFOCOM, 2011
  Proceedings IEEE}.\hskip 1em plus 0.5em minus 0.4em\relax IEEE, 2011, pp.
  2930--2938.

\bibitem{Yung2013}
J.~Jeong, Y.~Yi, J.~Cho, D.~Eun, and S.~Chong, ``{Wi-Fi} sensing: Should
  mobiles sleep longer as they age?'' in \emph{INFOCOM, 2013 Proceedings
  IEEE}.\hskip 1em plus 0.5em minus 0.4em\relax IEEE, 2013, pp. 2328--2336.

\bibitem{wang2009opportunistic}
W.~Wang, M.~Motani, and V.~Srinivasan, ``Opportunistic energy-efficient contact
  probing in delay-tolerant applications,'' \emph{IEEE/ACM Transactions on
  Networking (TON)}, vol.~17, no.~5, pp. 1592--1605, 2009.

\bibitem{zhou2010zifi}
R.~Zhou, Y.~Xiong, G.~Xing, L.~Sun, and J.~Ma, ``{ZiFi}: wireless lan discovery
  via zigbee interference signatures,'' in \emph{Proceedings of the sixteenth
  annual international conference on Mobile computing and networking}.\hskip
  1em plus 0.5em minus 0.4em\relax ACM, 2010, pp. 49--60.

\bibitem{wu2009footprint}
H.~Wu, K.~Tan, J.~Liu, and Y.~Zhang, ``Footprint: cellular assisted {Wi-Fi} ap
  discovery on mobile phones for energy saving,'' in \emph{Proceedings of the
  4th ACM international workshop on Experimental evaluation and
  characterization}.\hskip 1em plus 0.5em minus 0.4em\relax ACM, 2009, pp.
  67--76.

\bibitem{SwamiCR}
Y.~Chen, Q.~Zhao, and A.~Swami, ``Distributed spectrum sensing and access in
  cognitive radio networks with energy constraint,'' \emph{{IEEE} Trans. Signal
  Process.}, vol.~57, no.~2, pp. 783--797, Feb 2009.

\bibitem{azad2011optimal}
A.~P. Azad, S.~Alouf, E.~Altman, V.~Borkar, and G.~S. Paschos, ``Optimal
  control of sleep periods for wireless terminals,'' \emph{{IEEE} J. Sel. Areas
  Commun.}, vol.~29, no.~8, pp. 1605--1617, 2011.

\bibitem{kleinberg2004nearly}
R.~D. Kleinberg, ``Nearly tight bounds for the continuum-armed bandit
  problem,'' in \emph{Advances in Neural Information Processing Systems}, 2004,
  pp. 697--704.

\bibitem{auer2007improved}
P.~Auer, R.~Ortner, and C.~Szepesv{\'a}ri, ``Improved rates for the stochastic
  continuum-armed bandit problem,'' in \emph{Learning Theory}.\hskip 1em plus
  0.5em minus 0.4em\relax Springer, 2007, pp. 454--468.

\bibitem{BookPuterman}
M.~L. Puterman, \emph{Markov decision processes: discrete stochastic dynamic
  programming}.\hskip 1em plus 0.5em minus 0.4em\relax John Wiley \& Sons,
  2009, vol. 414.

\end{thebibliography}
